\documentclass[hidelinks, twoside]{article}

\usepackage[T1]{fontenc}
\usepackage[utf8]{inputenc}
\usepackage{palatino,epsfig}
\usepackage{lmodern}
\usepackage{parskip}
\usepackage{hyperref}
\usepackage{graphicx}
\usepackage{adjustbox}
\usepackage{natbib}
\usepackage[dvipsnames]{xcolor}
\usepackage[T1]{fontenc}
\usepackage{graphicx}
\usepackage{ifthen}
\usepackage{wrapfig}
\usepackage{multirow}
\usepackage{amsmath}
\usepackage{amssymb}
\usepackage{latexsym}
\usepackage{subcaption}
\usepackage{url}
\newcommand{\cross}{\times}

\newcommand{\union}{\cup}
\newcommand{\intersection}{\cap}

\newsavebox{\savepar}

\usepackage{amsmath}

\usepackage{hyperref}

\newtheorem{theorem}{Theorem}

\newtheorem{lemma}{Lemma}
\newtheorem{definition}{Definition}

\newenvironment{proof}[1][Proof]{\textbf{#1.} }{\ \rule{0.5em}{0.5em}}

\newcommand{\inpart}{\ensuremath{\Psi}}
\newcommand{\outpart}{\ensuremath{\Omega}}

\begin{document}

%\ecjHeader{x}{x}{xxx-xxx}{202x}{IOD Graph Informativeness under Crossover}{A.D. Pape, J.D. Schaffer, H. Sayama, C. Zosh}

\title{On the Preservation of Input/Output Directed Graph Informativeness under Crossover}
\author{Andreas Duus Pape, J. David Schaffer, Hiroki Sayama, Christoper Zosh}

\maketitle

\begin{abstract}
There exists a broad class of networks that connect inputs to outputs. These networks include chemical transformation networks, electrical circuits, municipal water systems, and neural networks. The goals of this paper are to provide a  theoretical foundation for  evolutionary crossover on this class of
graphs  and connect crossover to \textit{informativeness}, a measure of the connectedness of inputs to outputs. Informativeness is defined as: a \textit{partially informative} graph has at least one path from an input to some output, 
 a \textit{very informative} graph has a path from every input to some output, 
 and a \textit{fully informative} graph has  a path from every input to every output. If a neural network with non-zero weights and any number of layers is fully informative. As links are removed (assigned zero weight), it may become very, partially, or not informative. (The complement of informativeness is \textit{actionability}, which is a measure of how connected outputs are from inputs.)

We define a crossover operation on IOD Graphs
in which we find subgraphs with matching sets of forward and backward directed links to ``swap.''
With this operation, IOD Graphs can be subject to evolutionary computation methods. 
We show that fully informative parents may yield a non-informative child. We also show that under certain conditions, 
crossover compatible,
partially informative parents yield partially informative children, and very informative input parents with partially informative output parents yield very informative children. However, even under these conditions, full informativeness may not be retained. Similar results hold for actionability.

\end{abstract}

% \begin{keywords}
% Neural Networks, 
% Genetic Algorithm, 
% Crossover,
% Input/Output,
% Connectedness
% \end{keywords}

\section{Introduction} % (fold)
\label{sec:introduction}

There is a broad class of networks which connect inputs to outputs, used for modeling and problem solving in a variety of domains. Examples include chemical reaction networks \citep[e.g.,][]{unsleber2020exploration, wen2023chemical},  municipal water systems \citep[e.g.,][]{shinstine2002reliability, wu2009evolving},  data flow networks \citep[e.g.,][]{meng2004characterizing}, and electrical circuits \citep[e.g.,][]{koza1997reuse, shanthi2009practical}. We call these graphs Input/output Directed Graphs or IOD Graphs.  A common IOD Graph is the multi-layer perceptron, which is a feed-forward IOD Graph that converts inputs into ``useful'' outputs. Perceptrons are used to solve a variety of problems; for example, in cognitive science they are used to emulate human categorization behavior.\footnote{E.g., the classification learning literature beginning with \cite{SHJ1961} and including the neural network ALCOVE \citep{kruschke1992alcove}. See more below.} 

While evolutionary operations like crossover have been used in all these domains, the application of crossover to IOD Graphs remains sparse. We provide a theoretical foundation for crossover across this class of networks. First, we establish how crossover can be applied to any two ore more such networks.
We also connect crossover to informativeness, a measure of the connectedness of inputs to outputs, and provide proofs putting bounds on how informativeness flows from parent networks to their children.

First, we formally define this class of graphs as Input/Output Directed Graphs or IOD Graphs:
An \textit{IOD Graph} is a graph with a set of nodes $N$ and directed edges $E$, where $N$ contains 
(a) a set of ``input nodes'' $I \subset N$, where each $i \in I$ has no incoming edges and any number of outgoing edges, and (b) a set of ``output nodes'' $O \subset N$, where each $o \in O$ has no outgoing edges and any number of incoming edges,
and $I\intersection O = \emptyset$. Nodes $n\in N$ which are neither inputs nor outputs, so $n\notin I, n\notin O$, are called ``intermediate nodes.''

Because of the prominence of neural networks, we will generally describe IOD Graphs as information flow networks and interpret the properties and results from that perspective. However, as mentioned above, these results about crossover could also apply to e.g. municipal water system design. Also, while many neural networks are feed-forward, meaning the networks do not contain any feedback loops, we do not restrict our study or the definition of IOD Graphs to feed-forward networks. 

IOD Graphs solve problems by converting inputs into outputs, so an IOD Graph's ability to solve problems relies in part on the connectedness of inputs to outputs. We introduce a measure of the connected paths from inputs to outputs in `informativeness.' \textit{Informativeness} is a characterization of how information flowing into to the system is utilized, in the sense that it characterizes how many inputs are eventually connected to an output. More precisely, the \textit{informativeness} of an IOD Graph is a characterization of how many paths exist from inputs to outputs.  Informativeness is a key component of the value of an IOD Graph to solve a problem. 
At one extreme, an IOD Graph that has no path from an input to an output is useless at solving problems because it cannot deliver any information from the inputs to the outputs. On the other hand, an IOD Graph could have a path for every input and output pair, plus additional paths that interact, for example, as do the nodes in the layers of a multilayer perceptron.\footnote{Informativeness is also applicable for these other types of IOD Graphs. For example, a municipal water system which has no paths from inputs to outputs would function very poorly as water could not flow.}
(We also, as a corollary, introduce the concept of \emph{actionability}, which is the symmetric measure of how many outputs are eventually connected from an input.)

One problem appropriate for IOD Graphs is the classification learning problem in cognitive science/psychology. A canonical version of this problem was by \citet{SHJ1961}. In their laboratory, 
they showed human subjects objects characterized by a three-dimensional binary vector (e.g. big v. small, dark v. light, square v. triangle) and queried the category, A or B. 
The experimenter knew the true category, and told the human subject whether their guess was right; the error rate over time of different categorizations was measured (as the reader might imagine, some categorizations are easier to learn than others) and these experiments provide a data benchmark that future computational learning models strived to explain. One such explanation was ventured by \citet{kruschke1992alcove}, who approached this problem with a single-layer perceptron/feed-forward neural network with three inputs, two outputs, and nine nodes in the ``hidden layer,'' for a total of $14$ nodes. 
 The set of all IOD Graphs with $14$ nodes is vast: assuming at most one link between nodes, there are roughly $2^{132} \approx 5.4 \cross 10^{39}$ different possible IOD Graphs with that number of inputs, outputs, and intermediate nodes.
Growing potential solutions in this space via evolutionary methods seems wise. The DIVA model \citep{kurtz2007divergent} is another IOD Graph in this literature, which would model the SHJ problem as $3$ inputs and with two sets of $3$ outputs for a total of $6$ outputs. A crossover method like the one we describe in this paper could, for example, breed IOD Graphs from ALCOVE and DIVA implementations, possibly providing useful new solutions.

The purpose of this paper is to rigorously define a crossover operation and (begin to) characterize how informativeness and actionability changes under this crossover operation. That is, we seek to 
lay a strong theoretical foundation for crossover on a broader class of graphs, of which feed-forward neural networks, municipal water systems, electric circuits, etc., are special cases, and relate these different methods to informativeness, which we believe helps capture some of the effectiveness of IOD Graphs to process information. 

Since all neural networks are IOD Graphs, this paper is contribution to the substantial research effort applying crossover to neural networks.
Starting with the perceptron \citep{rosenblatt1957perceptron}, the vast majority of neural networks explored, including modern deep nets, are feed-forward networks. 
While more sophisticated computation is possible when networks include recurrent links,
there has yet to emerge a comprehensive theory of how the computational power of networks with recurrent links can be achieved. There has been much empirical research aimed at evolutionary computation for this task; for a recent summary, see \citet{stanley2019designing}.

The crossover process described here is applied on the graph directly and does not use evo-devo.
Evo-devo (or 
Evolutionary Development)
is a gene-driven process  in which a genome representation is translated into phenotypes, in this case, graphs, and crossover is done on the genomes, not the graphs themselves
  \citep{jacob1977evolution,gould1977ontogeny}.
This direct crossover on network structure may 
foster Lamarckian evolution and may be compatible with some evo-devo crossover methods
(see Section~\ref{sub:evo_devo}).

We formally define IOD Graphs and informativeness in Section~\ref{sec:definitions}.
We formally define the crossover operation in Section~\ref{sec:the_crossover_operation}, including the core definitions (Section~\ref{sub:crossover_related_definitions}) and a discussion of \emph{crossover compatibility} (Section~\ref{sub:crossover_compatibility}), which characterizes when two IOD Graphs can be subject to the crossover operation.
We apply this crossover operation to informativeness in Section~\ref{sec:the_preservation_of_informativeness}, `The Preservation of Informativeness,' which proves the core theorems. 
Section~\ref{sec:actionability} briefly extends the analysis from informativeness to a related idea of \emph{actionability}, which characterizes the amount of informed behavior flowing
out of a system,  in the sense that it characterizes how many outputs are connected from an input.
Section~\ref{sec:discussion} discusses aspects of the results, such as the competing conventions problem and the distribution of informativeness categories across networks of different degree.
Section~\ref{sec:conclusion} concludes.

\section{Definitions} % (fold)
\label{sec:definitions}

\begin{definition}\label{def:io-directed-graph}
	An \textit{Input/Output Directed Graph} or \emph{IOD Graph} is a graph with a set of nodes $N$ and directed edges $E$, such that $N$ contains two subsets:
	\begin{itemize}
	    \item     A set of nodes $I \subset N$, where each $i \in I$ has no incoming edges and any number of outgoing edges, and
	    \item     A set of nodes $O \subset N$, where each $o \in O$ has no outgoing edges and any number of incoming edges, and where $O \intersection I = \emptyset$
	\end{itemize}
\end{definition}
The set $I$ are inputs into the system and the set $O$ are outputs. We name the nodes which are neither inputs nor outputs as intermediate nodes. Specifically, $N\backslash \left(I\union O\right)$ is the set of intermediate nodes. We do not require the set of intermediate nodes to be non-empty.

Figure~\ref{fig:iodgraphexample1} depicts an example IOD Graph. 
A perceptron with any number of layers is an example of an IOD Graph; Figure~\ref{fig:perceptronexample} depicts a single-layer perceptron. Perceptrons are used to solve various problems relating input variables to output variables. We consider IOD Graphs as a larger class of networks that could plausibly be used to solve those same kinds of problems. 

One property we wish to define on IOD Graphs is how the inputs are connected to the outputs via directed paths.  We define \emph{informativeness}, which involves the connections via directed paths from the input nodes to the output nodes: 
 A \textit{partially informative IOD Graph} has at least one path from an input to an output, 
 a \textit{very informative IOD Graph} has a path from every input to an output, 
 and a \textit{fully informative IOD Graph} has a path from every input to every output. 
 On the other extreme, a \textit{non-informative IOD Graph} has no paths from inputs to outputs. A non-informative IOD Graph is worthless as a problem-solving engine.

\begin{figure}[h]
    \centering
    \begin{subfigure}[t]{0.45\textwidth}
        \centering
        \includegraphics[width=.9\textwidth]{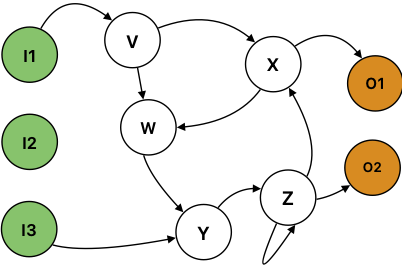}
        \caption{A Partially Informative IOD Graph}
		\label{fig:iodgraphexample1}
    \end{subfigure}%
    ~ 
    \begin{subfigure}[t]{0.45\textwidth}
        \centering
        \includegraphics[width=.9\textwidth]{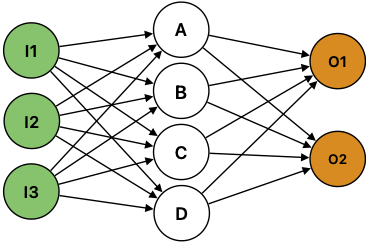}
        \caption{A Fully Informative IOD Graph\\(A Single-layer Perceptron)}
		\label{fig:perceptronexample}
    \end{subfigure}
    \caption{Two Input/Output Directed Graphs of Varying Informativeness}
\end{figure}
Figure~\ref{fig:iodgraphexample1} is a \emph{partially informative graph}: For example, $I1 \rightarrow V \rightarrow W \rightarrow Y \rightarrow Z \rightarrow O2$ is a directed path which leads from $I1 \in I$ to $O2 \in O$. However, it is not \emph{very informative}, because there is no path from $I2$ to any output, and it is not fully informative because very informative is a necessary condition for \emph{fully informative}. If the IOD Graph were modified by adding a directed edge from $I2 \rightarrow W$, then this modified IOD Graph would be very informative.
A multi-level perceptron with non-zero weights on its edges, for example depicted in Figure~\ref{fig:perceptronexample}, is fully informative. Every input has a directed path to every output.

%Incorrect: Every input has a directed path to every output. Removing even one of these links, however, would reduce the IOD Graph depicted in Figure~\ref{fig:perceptronexample} to be only very informative.

% subsection input_output_directed_graph_fundamental_definitions (end)

 \section{The Crossover Operation} % (fold)
 \label{sec:the_crossover_operation}
\subsection{Crossover-related Definitions} % (fold)
\label{sub:crossover_related_definitions}

We define a crossover operation on IOD Graphs
in which we find subgraphs with matching sets of forward and backward directed links to ``swap.'' Here we define the relevant terms to fully define the crossover operation. In the genetic algorithm,  the one point crossover operation 
splits the two parent genomes, and then
splices
the initial part of one parent's genome and 
the latter part of the other parent's genome.
Similarly,
we split two IOD Graphs into an input part and an output part
and then splice
 the input part of one IOD Graph to the output part of the other.
 
In order to define this operation, we must define \emph{first}, how to split an IOD Graph, and \emph{second,} how to splice them.

The natural way to split an IOD Graph is to partition its nodes into two non-overlapping sets, where all input nodes are in one part and all output nodes are in the other part.\footnote{Definition: A \emph{partition} $(S_1,S_2)$ of the set $S$ requires  that $S_1 \union S_2 = S$ and  $S_1 \intersection S_2 = \emptyset$.} 
We define an \emph{IO Partition} as follows:
\begin{definition}\label{def:partition}
Let $G$ be an Input/Output Directed Graph, with corresponding subsets of nodes $I$ and $O$. Then define an \emph{IO Partition} of $G$ to be a partition  $(\inpart{},\outpart{})$ of $N$, 
where $I \subseteq \inpart{}$ and $O \subseteq \outpart{}$.
\end{definition}
For any IO Partition $(\inpart{},\outpart{})$, we refer to $\inpart{}$ as the input part and $\outpart{}$ the output part of the partition.

Given this definition of an IO partition, we define
$M(\inpart{},\outpart{})$ as the set of all edges which connect the parts of the partition. This is called a \emph{cut} in network theory. Using $(n_1, n_2)$ to denote a directed edge from $n_1$ to $n_2$ and E to denote the list of all edges, we can formally define $M(\inpart{},\outpart{})$ as follows:

For all $e \in E$, where $e = (n_1,n_2)$ or $e = (n_2,n_1)$
	\begin{align*}
		M(\inpart{}, \outpart{}) &=  \left\{e \quad 
		\middle| \quad  n_1 \in \Psi, n_2 \in \Omega
	 \right\}
	\end{align*}

% Formally, using $(n_1, n_2)$ to denote a directed edge from $n_1$ to $n_2$,
%	\begin{align*}
%		M(\inpart{}, \outpart{}) &=  \left\{ e \in E \quad \middle| 
%		\begin{array}{rl}
%			e=&(n_1 , n_2)\\
%			 &\textrm{ or }           \\
%			e=&(n_2 , n_1)                         	
%		\end{array}
%		\textrm{ for all } n_1 \in \inpart{}, n_2\in \outpart{}
%	 \right\}
%	\end{align*}    
We call
$M(\inpart{},\outpart{})$ the \emph{separating membrane} associated with IO Partition $(\inpart{},\outpart{})$.
An IO Partition and corresponding separating membrane are depicted in Figure~\ref{fig:IOPartitionMembraneExample}.

%%%%%%%%%%%%%%%%%%%%IOPARTITION AND MEMBRANE EXAMPLE
 \begin{figure}[h]
     \centering  \includegraphics[width=.60\textwidth]{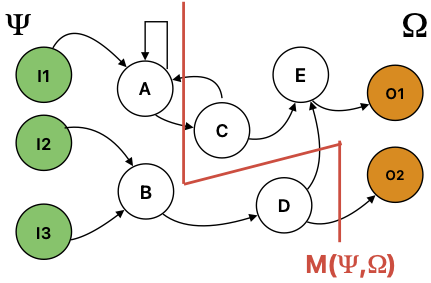}
     \caption{An IO Partition and corresponding membrane\\
     \label{fig:IOPartitionMembraneExample}}
 \end{figure}%
Each $e\in M(\inpart{},\outpart{})$ either connects from the input part to the output part or the reverse. If it connects from the input part to the output part, we call it a \emph{forward link}, because it links from inputs to outputs. If instead it connects from the output to the input, we call it a \emph{backwards link}. Let $F(\inpart{},\outpart{})$ be the set of forward links in a membrane $M(\inpart{},\outpart{})$ and $B(\inpart{},\outpart{})$ be the set of backwards links in $M(\inpart{},\outpart{})$. Clearly, $M(\inpart{},\outpart{}) = F(\inpart{},\outpart{}) \union B(\inpart{},\outpart{})$.

\begin{figure}[h]
	\centering
   \includegraphics[width=.55\textwidth]{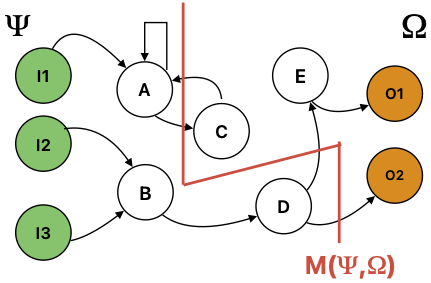}
	\caption{This IO Partition is input-contiguous\\
 but not output-contiguous}
	\label{fig:IOPartitionContiguous}
\end{figure}
We also wish to have a way to describe a partition as cohering within the input or output parts.
\begin{definition}\label{def:contigous_partitions}
		For an IO Partition $(\inpart{},\outpart{})$:
		\begin{itemize}
			\item \emph{The input part $\inpart{}$ is contiguous} (or, equivalently, \emph{$(\inpart{},\outpart{})$ is input-contiguous}) if, for every $n\in \inpart{}\backslash I$, there is a path from some input $i\in I$ to $n$ which lies entirely within $\inpart{}$, and
			\item \emph{The output part $\outpart{}$ is contiguous} (or, equivalently, \emph{$(\inpart{},\outpart{})$ is output-contiguous}) if, for every $n' \in \outpart{}\backslash O$, there is a path from $n'$ to some output $o\in O$ which lies entirely within $\outpart{}$.
		\end{itemize}
		
		If both the input part \inpart{} and output part \outpart{} are contiguous, we say \emph{the IO partition $(\inpart{},\outpart{})$ is contiguous.} 
\end{definition}
Requiring a partition to be contiguous means intermediate nodes in the input part are connected from inputs and intermediate nodes in the output part are connected to outputs.
Figure~\ref{fig:IOPartitionContiguous} depicts an IOD Graph and IO Partition which is input-contiguous but not output-contiguous. Output contiguousness fails because there is no path from $C$ to an output which is contained in \outpart{}.

\subsection{Crossover compatibility} % (fold)
\label{sub:crossover_compatibility}

Suppose $G, G'$ are IOD Graphs with IO Partitions $(\inpart{},\outpart{})$ and $(\inpart{}',\outpart{}')$. Then we say these IO Partitions are \emph{crossover compatible} if they have the same number of forward and backward links, and the rather technical assumption (typically easily satisfied) that the input part of one parent shares no nodes with the output part of the other and vice versa.
\begin{definition}\label{def:crossover_compatible}
	IOD Graphs and IO Partitions $\left\{G,(\inpart{},\outpart{}) \right\}$ and $\left\{G',(\inpart{}',\outpart{}') \right\}$ are \emph{crossover compatible} if $\inpart{} \intersection \outpart{}' = \emptyset$,  $\inpart{}' \intersection \outpart{} = \emptyset$, and
		\begin{align*}
			\left|F(\inpart{},\outpart{})\right| &=  \left|F(\inpart{}',\outpart{}')\right|, \textrm{and} \\
				\left|B(\inpart{},\outpart{})\right| &=  \left|B(\inpart{}',\outpart{}')\right|
		\end{align*}
\end{definition}

If IO Partitions $(\inpart{},\outpart{})$ and $(\inpart{}',\outpart{}')$ are crossover compatible, then they can be used to create crossover children. For simplicity, and without lack of generality, we define the crossover child constructed from crossover compatible $\inpart{}$ and $\outpart{}'$. In this case, we call $G$ the input parent, because it contributes the input part, and $G'$ the output parent, because it contributes the output part.%
\footnote{Notably, there may also be a crossover child that can be constructed from $(\inpart{}',\outpart{})$. The process and discussion is identical with the roles of $\inpart{}$ and $\inpart{}'$, and $\outpart{}'$ and $\outpart{}$, reversed.}
 First, we define a crossover membrane $\hat{M}$:
\begin{definition}\label{def:crossover_membrane}
	Suppose $G, G'$ are IOD Graphs with crossover compatible IO Partitions $(\inpart{},\outpart{})$ and $(\inpart{}',\outpart{}')$. Then a $(G,G')$ crossover membrane $\hat{M}$ is defined as as all edges $f''$ and $b''$, where
	\begin{enumerate}
\item For each forward edge $f \in F(\inpart{},\outpart{})$, select, without replacement, a forward edge $f'$ in $F(\inpart{}',\outpart{}')$, and then construct the edge $f''$ which connects the source of $f$ to the destination of $f'$, and
\item For each backward edge $b \in B(\inpart{},\outpart{})$, select, without replacement, a backward edge $b'$ in $B(\inpart{}',\outpart{}')$, and construct the edge $b''$ which connects the source of $b'$ with the destination of $b$.
	\end{enumerate}
\end{definition}
Note that a given pair of crossover compatible IO Partitions $(\inpart{},\outpart{})$ and $(\inpart{}',\outpart{}')$, there must exist  at least one crossover membrane, as shown in Lemma~\ref{lem:CrossoverMembraneExists}.
There may be, and often are, more than one possible crossover membrane. In that case, $\inpart{}$ and $\outpart{}'$ can be connected in meaningfully different configurations.%
%

% 8. The authors should justify why a crossover membrane always exists. Currently, it is only stated that this is the case (below Definition 5 on page ).
% ADP: Adding new proof
\begin{lemma}\label{lem:CrossoverMembraneExists}
 Suppose  $G, G'$ are two crossover compatible IOD Graphs. Then there exists a $(G,G')$ crossover membrane.
\end{lemma}
\begin{proof}
Since $G, G'$ are crossover compatible, then select crossover compatible IO Partitions $(\inpart{},\outpart{})$ for $G$ and $(\inpart{}',\outpart{}')$ for $G'$.

	Then define $f,f'$ and $b,b'$ as the sets of forward and backward links; namely let $f,f'$ be $f = F(\inpart{},\outpart{})$ and $f' =F(\inpart{}',\outpart{}')$; and let $b,b'$ be $b = B(\inpart{},\outpart{})$ and $b' =B(\inpart{}',\outpart{}')$.

	By definition of crossover compatible, $|f|=|f'|$ and $|b|=|b'|$. Create two one-to-one and onto mappings $g: f \rightarrow f'$ and $h: b' \rightarrow b$.

\newcommand{\linksource}[1]{\ensuremath{{#1}_{\textrm{source}}}}
\newcommand{\linkdest}[1]{\ensuremath{{#1}_{\textrm{dest}}}}	

	For a link $l$, define \linksource{l} as the source node of $l$ and let \linkdest{l} as the destination node of $l$; so by definition each link $l = (\linksource{1},\linkdest{l})$.

    		Then define $\hat{M}_f$, $\hat{M}_b$ as:
	\begin{align*}
	\hat{M}_f &= \left\{ 
		    (\linksource{x}, \linkdest{g(x)}) \quad \middle| \quad x \in f
		\right\} \\
	\hat{M}_b &=
		\left\{ 
		    (\linksource{y}, \linkdest{h(y)}) \quad \middle| \quad y \in b'
		\right\} \\
	\end{align*}
		$\hat{M}_f \union \hat{M}_b$ is a $(G,G')$ crossover membrane.%
\end{proof}%

\bigskip

Now we define the child of crossover:
%%%%%%%%%%%%%%%%%%%%%%%%%%%%%%%%%%%%%%%%%%%%%%%%
\begin{definition}\label{def:crossover_child}
	Suppose $G, G'$ are IOD Graphs with crossover compatible IO Partitions $(\inpart{},\outpart{})$ and $(\inpart{}',\outpart{}')$. Then construct a crossover child $C$ as the graph consisting of the following nodes and edges:
	\begin{itemize}
		\item The set of nodes in $C$ is $\inpart{} \union \outpart{}'$
		\item The set of edges in $C$ is the union of:
		\begin{itemize}
			\item the set of edges $e$ which connect elements of $\inpart{}$, 
			\item the set of edges $e'$ which connect elements of $\outpart{}'$, and
			\item a crossover membrane $\hat{M}$
		\end{itemize}
 	\end{itemize}
\end{definition}

With this definition in hand,
we can define \emph{IOD Graph crossover}
\begin{definition}\label{def:IODGraphCrossover}
	IOD Graph Crossover is any algorithm which transforms any two IOD Graphs with crossover compatible IO Partitions
	into a crossover child.
\end{definition}

Crossover is depicted in Figure~\ref{fig:IODGraphCrossover}.
On the left, we have the Input Parent IOD Graph. 
The set of nodes to the left of $M1$ are the set $\inpart{}$, i.e. 
$\inpart{} = \left\{I1, I2, I3, A, B, C\right\}.$  The set of nodes to the right are the set $\outpart{} = \left\{ C, E, O1, O2\right\}$.  The membrane $M1=M(\inpart{},\outpart{})$ contains three forward links and one backward link: 
$F(\inpart{},\outpart{})=\left\{(A,C), (D,E), (D,O2)\right\}$, $B(\inpart{},\outpart{}) = \left\{(C,A)\right\}$

\begin{figure}[h]
    \centering
    \includegraphics[width=.68\textwidth]{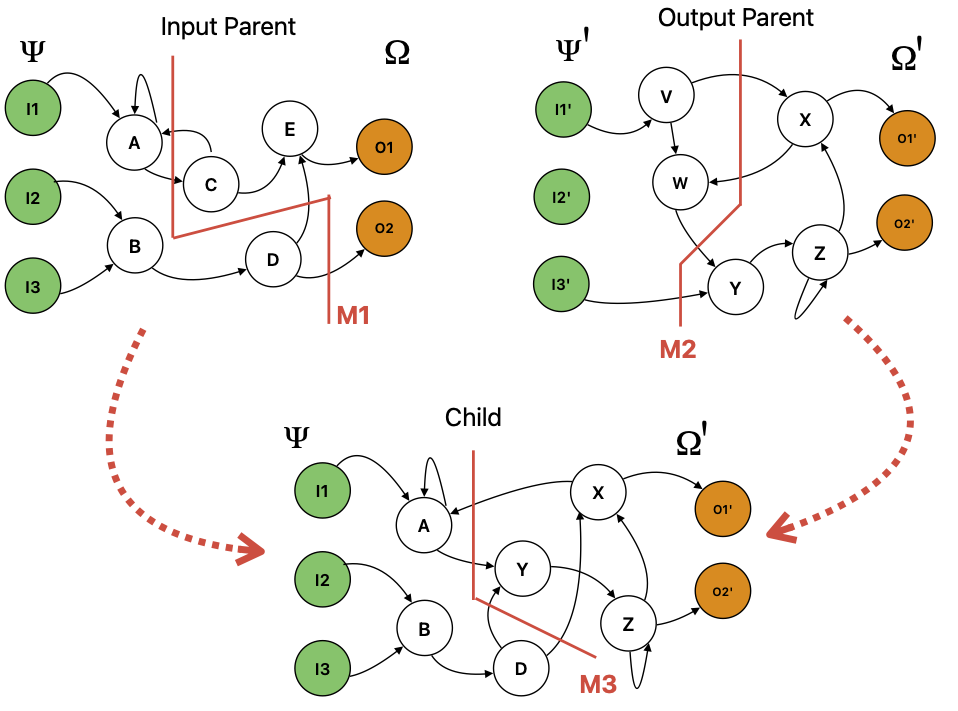}
    \caption{The Crossover Operation on Input/Output Directed Graphs}
    \label{fig:IODGraphCrossover}
\end{figure}%

On the right, we have the Output Parent IOD Graph. $\inpart{}' = \left\{ I'1, I'2, I'3, V, W \right\}$ and $\outpart{}' = \left\{ X, Y, Z, O1, O2\right\}$.
Like $M(\inpart{},\outpart{})$, $M(\inpart{}',\outpart{}')$ contains three forward links and one backward link:
$F(\inpart{}',\outpart{}')=\left\{(V,X), (W,Y), (I3,Y)\right\}$, $B(\inpart{}',\outpart{}') = \left\{(X,W)\right\}$.
The Input Parent and Output Parent are crossover compatible because the membranes $M1=M(\inpart{},\outpart{})$ and $M2=M(\inpart{}',\outpart{}')$ have the same number of forward links (three) and the same number of backward links (one). Since they are crossover compatible, we can construct a crossover membrane $M3$ which connects the input part of the Input Parent (set $\inpart{}$) to the output part of the Output Parent (set $\outpart{}'$). In Figure~\ref{fig:IODGraphCrossover}, we depict crossover membrane $M3 = \left\{ (A,Y), (D,Y), (D,X), (X, A) \right\}$. As with $M(\inpart{},\outpart{})$ and $M(\inpart{}',\outpart{}')$, $M3$ contains three forward and one backward link. Given this crossover membrane, the Child is a well-defined IOD Graph.

Note in this example, the crossover child $C$ is an IOD Graph. This is always true; that is, the set of IOD Graphs is closed under crossover. The proof of this claim is  straightforward. We state this formally in Lemma~\ref{lem:IODGraphsClosedUnderCrossover}:
\begin{lemma}\label{lem:IODGraphsClosedUnderCrossover}
	If $G$ and $G'$ are IOD Graphs, then any child produced by crossover will also be an IOD Graph.
\end{lemma}
\begin{proof}\label{pf:IODGraphsClosedUnderCrossover}
		 Definition~\ref{def:crossover_child} implies that the crossover child $C$ is a graph (that is, a set of nodes and a set of edges). Moveover, it contains a set of inputs $I\subseteq \inpart{}$, where all nodes $i\in I$ have no incoming links,
		  and a set of outputs $O \subseteq \outpart{}$, where are notes $o\in O$ have no outgoing links. 
		  Since $\inpart{} \intersection \outpart{} = \emptyset$, then $I \intersection O = \emptyset$. Therefore the child $C$ is an IOD Graph. \end{proof}

The reader may be interested in applications to feed-forward networks such as neural networks or perceptrons. Like the set of IOD Graphs, the set of feed-forward networks is closed under crossover, as shown in Lemma~\ref{lem:FeedForwardIODGraphsClosedUnderCrossover}:
\begin{lemma}\label{lem:FeedForwardIODGraphsClosedUnderCrossover}
		If $G$ and $G'$ are feed-forward IOD Graphs, then any child produced by crossover will also be a feed-forward IOD Graph.
\end{lemma}
\begin{proof}\label{pf:lem:FeedForwardIODGraphsClosedUnderCrossover}
	Suppose two feed-forward IOD Graphs $G$ and $G'$ produce a child which is not feed-forward. Then there is an edge which causes a loop in a child which was not present in either parent. This loop cannot involve only nodes within \inpart{} nor nodes only within \outpart{}', because $G$ and $G'$ are feed-forward networks. This loop must have been ``created'' by the crossover. This means it must contain one forward and one backward link from the crossover membrane. However, the crossover membrane contains no backward links. Contradiction.
\end{proof}

While the set of feed-forward networks is closed under crossover, the set of multilayer perceptrons is not; crossover can produce children who do not maintain the multilayer structure.
Crossover between perceptrons \emph{can} be guaranteed to produce perceptrons if the set of IO Partitions is restricted such that all nodes in the same layer are in the same part of the partition (e.g. if a node $n\in \inpart{}$ then all $\tilde{n}$ in the same layer as $n$ are also in $\inpart{}$). This proof is omitted.

\section{The Preservation of Informativeness} % (fold)
\label{sec:the_preservation_of_informativeness}

There is no guarantee that informativeness of two IOD Graphs are preserved under crossover.%, as shown in Theorem~\ref{thm:no_preservation}.
 
 	\begin{figure}[h]
 	    \centering
 	    \includegraphics[width=.5\textwidth]{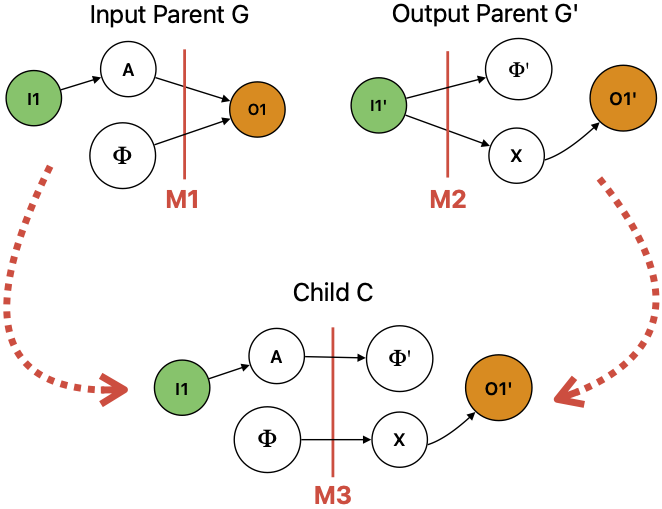}
 	    \caption{Two fully informative parents\\may yield a non-informative child}
 	    \label{fig:IODGraphFullyToNonInform}
 	\end{figure}%
 Figure~\ref{fig:IODGraphFullyToNonInform} illustrates an example with two fully informative parents that each have one input and one output, and a path which connects them. In addition, input parent $G$ has what we call a `false input' $\phi$ and the output parent has a corresponding `false output,' $\phi'$. The figure illustrates how a false input/false output pair can completely disable a path from an input to an output in the child. The crossover operation acts like flipping a switch, breaking flow of information.\footnote{Note: there is an alternative crossover membrane which would preserve full informativeness: namely the links $(A,X)$ and $(\phi, \phi')$ comprise a crossover membrane which results in a fully informative child.}$^,$\footnote{Later, we say that $\phi$ and $\phi'$ are \emph{dangling nodes}. See the \emph{no dangling nodes condition} below (Definition~\ref{def:no_dangling_nodes}).}

The false output ends the flow of information from the input. This `causes' non-informativeness. However, a false output alone is not sufficient for this example. The false input is also required. Without a false input in Input Parent $G$, these IO Partitions would not be crossover compatible and therefore crossover could not be performed. In particular, $\left|F(\inpart{},\outpart{})\right|$ would be two, while $\left|F(\inpart{}',\outpart{}')\right|$ would be one.
Node $X$ would require an incoming link from the membrane but there would be no link to be found.

 The process shown in Figure~\ref{fig:IODGraphFullyToNonInform} generalizes exponentially with inputs and outputs, as shown in
 Theorem~\ref{thm:no_preservation}:
 
%%%%%%%%%%%%%%%%%%%%%%%%%%
%%% NO INFORMATIVENESS M\inpart{}ST BE PRESERVED THEOREM %%%
\begin{theorem}\label{thm:no_preservation}
	The child of two IOD Graphs may retain no informativeness of the parents. 
\end{theorem}
\begin{proof}\label{pf:no_preservation}
	We show a child of two IOD Graphs may retain no informativeness of the parents by creating two IOD Graphs of arbitrary informativeness, then showing they can produce a crossover child which has no informativeness.
	
	Suppose $G$ and $G'$ are IOD Graphs which each have $J$ inputs $I,I'$ and $K$ outputs $O,O'$.	
	$G$ is the input parent and $G'$ the output parent.
	
	% G
	Suppose $G$ has the following structure:
	
	For each  input node $i\in I$, there are $K$ paths, called $p(i,k)$, which may have intermediate nodes in common. Path $p(i,k)$ leads from input $i$ to either output $k$ or an intermediate node with no outgoing links. If every path $p(i,k)$ leads from input $i$ to output $k$, $\forall i,k$, then the IOD Graph is fully informative. The more paths end in intermediate nodes, the lower the informativeness. Therefore, any level of informativeness can be achieved by the choice of paths $p(i,k)$ in the construction of $G$.  
	
	In addition to these paths, $G$  has $J\cdot K$ nodes which we will call ``false inputs.'' We name these nodes $\phi(i,k)$, for $i=1, \ldots, J, k=1\ldots K$. 
	False input $\phi(i,k)$ has no incoming edges and one outgoing edge, which connects directly to output $k$. $\phi(i,k)$  is the false input we will use to replace input $i$.
	
	$G$ has no other nodes or edges other than those described above.
	
	The IO Partion $(\inpart{},\outpart{})$ associated with $G$ is: %
	\begin{align*}
	(\inpart{} = N \backslash O,& \outpart{} = O)
	\end{align*}

	%G'
	Suppose $G'$ has the following structure:
	
	For each  input node $i'\in I'$, there are $K$ paths, called $p'(i',o')$. Path $p'(i',o')$ leads from input $i'$ to either output $o'$ or an intermediate node with no outgoing links. As above, these paths may have nodes in common. Like in graph $G$ on the input part, IOD Graph $G'$ can achieve any level of informativeness.
	
	$G'$ has $J\cdot K$ nodes which we will call ``false outputs.'' We name these nodes $\phi'(i',o')$, for $i'\in I', o'\in O'$. False output $\phi'(i',o')$ has one incoming edge and no outgoing edges. It is connected directly from input $i'$. $\phi'(i',o')$  is the false output we will use to replace output $o'$.

	Like IOD Graph $G$, $G'$ has no other nodes or edges other than those described above.
	
	The IO Partion $(\inpart{}',\outpart{}')$ associated with $G'$ is:  %
	\begin{align*}
	(\inpart{}' = I,& \outpart{}' = N\backslash I)
	\end{align*}
	
	We then construct a crossover membrane $\hat{M}$ which breaks all informative paths by connecting potential informative paths to false inputs and outputs:
	
	First we create a one-to-one and onto mapping from $I$ to $I'$, called $j: I \rightarrow I'$,	
	and a one-to-one and onto mapping $O$ to $O'$ called $k: O\rightarrow O'.$
	
	Second, for each $i\in I$ and $o\in O$, $\hat{M}$ does the following:
	\begin{itemize}
		\item it connects path $p(i,o)$ to false output $\phi'\left(j(i),k(o)\right)$, and
		\item it connects each false input $\phi(i,k)$ to path $p'(j(i),k(o))$.
	\end{itemize} 

Then in the resulting child, every path leading out of each input $i\in I$ leads to false output $\phi'(\cdot)$. This means that the IOD Graph has no informativeness.\end{proof}%%%%%%%%%%%%%%% END PROOF

\begin{figure}[h]
    \centering
    \begin{subfigure}[t]{0.45\textwidth}
        \centering
        \includegraphics[width=.9\textwidth]{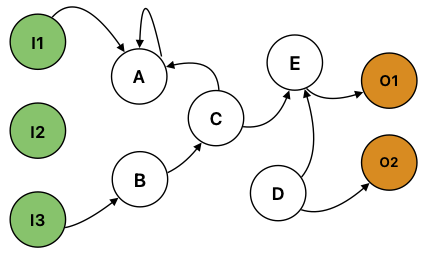}
        \caption{Nodes $A$ and $D$ Dangle}
		\label{fig:dangles}
    \end{subfigure}
	~
    \begin{subfigure}[t]{0.45\textwidth}
        \centering
        \includegraphics[width=.9\textwidth]{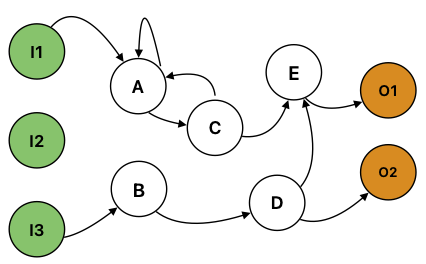}
        \caption{No Dangling Intermediate Nodes}
		\label{fig:nodangles}
    \end{subfigure}%
    \caption{The No Dangling Nodes Condition}
\end{figure}

False inputs and false outputs are what we call `dangling nodes.' Dangling nodes are intermediate nodes which do not lie on a path from an input to an output.\footnote{For example, in Figure~\ref{fig:IODGraphFullyToNonInform},
which illustrates  full informativeness not being preserved, nodes $\phi$ and $\phi'$ dangle.}
To preserve informativeness through inheritance, we use the property which rules out dangling nodes, the \emph{no dangling nodes condition}:
%
% NO DANGLING NODES CONDITION
%
\begin{definition}\label{def:no_dangling_nodes}
	An IOD Graph satisfies the \emph{no dangling nodes condition} if every intermediate node is on a path from an input $i\in I$ to an output $o \in O$.
\end{definition}
In Figure~\ref{fig:dangles}, both nodes $A$ and $D$ are ``dangling'' because neither is on a path from an input to an output. Figure~\ref{fig:nodangles} satisfies the No Dangling Node condition because every intermediate node falls on a path from an input to an output. Node $I2$ is not on a path to an output, but it is an input, not an intermediate node, so does not violate the condition.

It follows immediately that IOD Graphs which satisfy the No Dangling Nodes condition and have a non-empty set of intermediate nodes are partially informative. This is because any existing intermediate node must lie on a path from an input to an output, and therefore such a path exists. However, an IOD Graph can satisfy the no dangling nodes condition and not be very informative, because there can be inputs which connect to nothing, as seen in Figure~\ref{fig:nodangles}.

Preservation of the No Dangling Nodes condition under inheritance is the cornerstone of informativeness preservation. It turns out that \emph{contiguousness} of parents' IO Partitions is enough to guarantee that the no dangling condition is inherited by the child. (In particular, the input part must be contiguous and the output part must be contiguous.)
Figure~\ref{fig:IODGraphThmNoDangle} illustrates an example of why this property delivers inheritance of the No Dangling Nodes condition. Below, Theorem~\ref{thm:no_dangling_nodes_preserved} states the claim formally and proves it.
	\begin{figure}[h]
	    \centering
	    \includegraphics[width=.68\textwidth]{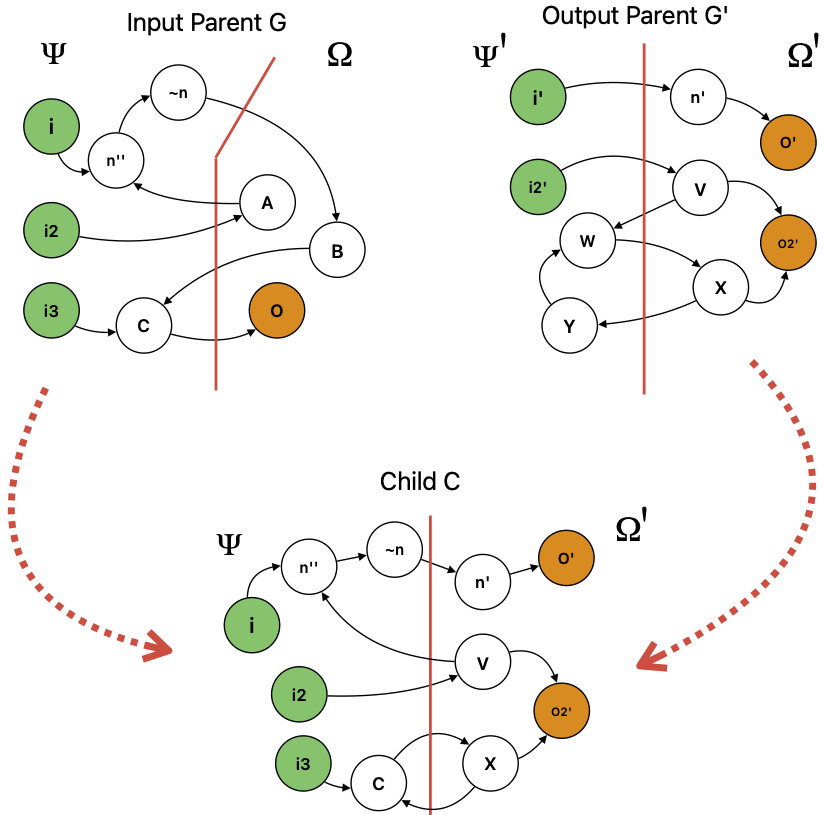}
	    \caption{No Dangling Nodes Condition Inheritance: An Example}
	    \label{fig:IODGraphThmNoDangle}
	\end{figure}%

In Figure~\ref{fig:IODGraphThmNoDangle}, suppose that the node $n''$ is an arbitrary node in the input part \inpart{}. We seek to show that $n''$ must be on a path from an input to an output in the child $C$. In this example, that path turns out to be $i \rightarrow n'' \rightarrow \sim n \rightarrow n' \rightarrow o'$. We build that path in pieces: pieces $p$, $q$, and $r$, which we define below.

 Consider the input parent $G$. Since $G$ is input-contiguous, we are assured of the path from input $i$ to node $n''$, which we call path $p$. Path $p$
	 lies entirely within \inpart{}. Since $G$ satisfies the No Dangling Nodes condition, we are also assured of a path from $i2$ to $O$. In this case, the path is $i2 \rightarrow A \rightarrow n'' \rightarrow \sim n \rightarrow B \rightarrow C \rightarrow O2$.
     This path was selected to illustrate that the path may cross between \inpart{} and \outpart{} multiple times; in particular, we are not assured from No Dangling Nodes alone that the subpath from $i2$ through $\sim n$ to $n''$ resides entirely within \inpart{}.

     There is a part of the path $i2 \rightarrow A \rightarrow n'' \rightarrow \sim n \rightarrow B \rightarrow C \rightarrow O2$ which we use later: the subpath  $n'' \rightarrow \sim n$, which we call path $q$.
	
	Now consider the output parent $G'$ in Figure~\ref{fig:IODGraphThmNoDangle}. Here, the No Dangling Nodes condition assures us that this $n' \in \outpart{}'$ that has been selected must have some path to some output $o'$ that lies entirely within \outpart{}', which we call path $r$.\footnote{Strictly speaking, No Dangling Nodes is not required for the Output Parent $G'$ in this part of the proof. It is required later, when we consider $n'' \in \outpart{}'$.}
	
Note that the IO Partitions of $G$ and $G'$ depicted in Figure~\ref{fig:IODGraphThmNoDangle} are crossover compatible because
$\left|F(\inpart{},\outpart{})\right| = \left|F(\inpart{}',\outpart{}')\right| = 3$ and $\left|B(\inpart{},\outpart{})\right| = \left|B(\inpart{}',\outpart{}')\right| = 2$.

	Consider now Child $C$, we see the path $p \rightarrow q \rightarrow r$ which is $i \rightarrow n'' \rightarrow ~n \rightarrow n' \rightarrow o'$. This is a path which satisfies that node $n''$ lies on a path from an input to an output in Child $C$. This demonstrates that No Dangling Nodes is inherited.

%%%%%%%%%%%%%%%%%%%%%%%%%%%%%%%%%%%%%%%%%%
%%% NO DANGLING NODES PRESERVED THEOREM %%
\begin{theorem}\label{thm:no_dangling_nodes_preserved}
	Suppose IOD Graph $G$ has input-contiguous IO partition $(\inpart{},\outpart{})$ and IOD Graph $G'$ has output-contiguous IO partition $(\inpart{}',\outpart{}')$, and further suppose both graphs $G$ and $G'$ satisfy the no dangling nodes condition.
	Then any crossover child produced by $(G,(\inpart{},\outpart))$ and $(G',(\inpart{}',\outpart{}'))$  must also satisfy the no dangling nodes condition.
\end{theorem}

\begin{proof}\label{pf:no_dangling_nodes_preserved}
	Suppose $(G,(\inpart{},\outpart{}))$ and $(G',(\inpart{}',\outpart{}'))$ are crossover compatible. Let the nodes of $G$ include input nodes $I$ and output nodes $O$ and the nodes of $G'$ include of input nodes $I'$ and output nodes $O'$. 
	
	Let $\hat{M}$ be a crossover membrane that is used to construct a crossover child $C$. 
Now, all nodes of $C$ must belong to $\inpart{}$ or $\outpart{}'$. 

Suppose there are no intermediate nodes in $C$. Then the result is immediate.

Now suppose there is at least one intermediate node in $C$. Consider an arbitrary intermediate node in $C$.

\textbf{Case 1.} Suppose that intermediate node $n''\in \inpart{}\backslash I$. 

 Since the IO Partition $(\inpart{},\outpart{})$ is contiguous, there must be a path from some $i\in I$ to $n''$ which is contained within $\inpart{}$. We call this path $p=(i, \ldots, n'')$.
	
	Now, because $G$ satisfies the no dangling nodes condition, $n''$ in $G$ must be on a path from $(\tilde{i}, \ldots, n'', \ldots, \tilde{o})$ for some $\tilde{i}\in I, \tilde{o} \in O$.\footnote{Note that the initial segment of this path $(\tilde{i},\ldots,n'')$ need not lie within \inpart{}, which is why the earlier path segment $p$ is needed.}
	Consider the subpath $(n'',\ldots, \tilde{o})$. Now that path must contain a node in $\outpart{}$. Find the first node in the path $(n'',\ldots, \tilde{o})$ which is in $\outpart{}$, and consider the node $\tilde{n}$ which links to that node.  That is the path $(i, \ldots, n'', \ldots, \tilde{n})$. We define the subpath $q=(n'', \ldots, \tilde{n})$. Note that $q$ resides entirely within $\inpart{}$. This means the path $p \rightarrow q$ resides entirely within $\inpart{}$.
	
	In $C$, $\tilde{n}$ must link to some $n'\in \outpart{}'$ by construction of the crossover membrane $\hat{M}$.  Since the IO partition $(\inpart{}',\outpart{}')$ is contiguous, there must be a path from $n'$  to an output $o'\in O'$ fully contained in $\outpart{}$, called path $r=(n', \ldots, o')$.  
	Since $p$, $q$, and $r$ reside in $C$, then the overall path %
	\begin{align*}
	p \rightarrow q \rightarrow r &= (i, \ldots, n'', \ldots, \tilde{n}, n', \ldots, o' )
	\end{align*}
	must be in $C$ and therefore the no dangling node condition is satisfied for $n''$. 
	
	Note: It is possible in the previous paragraph that $n'$ may be an output, i.e. $n' \in O'$. In this case the corresponding path $R$ is the trivial $r=(n')$ and the rest of the proof follows: path $(i,\ldots, \tilde{n},n')$ is a path which is in $C$ and, since $n'\in O'$, it is a path which connects inputs to outputs.	
	
\textbf{Case 2.} Suppose that the intermediate node $n''\in \outpart{}\backslash O'$

A symmetric argument applies, stated here for completeness:
	Since the IO Partition $(\inpart{}',\outpart{}')$ is contiguous, there must be a path from $n''$ to some $o'\in O'$ which is contained within $\outpart{}'$. We shall call this path $p'=(n'', \ldots, o')$.
	 Now, $n''$ in $G'$ is on a path from $(\bar{i}, \ldots, n'', \ldots, \bar{o})$ for some $\bar{i}\in I', \bar{o} \in O'$. Consider the path $(\bar{i},\ldots, n'')$. Now that path must contain a node in $\inpart{}'$. Find the last node in the path $(\bar{i},\ldots, n'')$ which is in $\inpart{}$, and consider the node $\bar{n}$ which links from that node. This is the path $(\bar{n}, \ldots, n'', \ldots, \bar{o})$. We name the subpath $q'=(\bar{n}, \ldots, n'')$. Note that $q'$ resides entirely within $\outpart{}'$ and therefore $q'\rightarrow p'$ also resides entirely within $\outpart{}$.
 In $C$, there must be some some $n\in \inpart{}$ which links to $\bar{n}$. By the IO partition $(\inpart{},\outpart{})$ being contiguous, there must be a path an input $i\in I$ to $n$ fully contained in $\inpart{}$, which we call $r'=(i, \ldots, n)$. Therefore the path %
\begin{align*}
 r'\rightarrow q' \rightarrow p' &= (i, \ldots, n, \bar{n}, \ldots, n'', \ldots, \bar{o} )
\end{align*}
must be in $C$ and therefore the no dangling node condition is satisfied for $n''$.

\medskip
	
Given Case 1 and Case 2, we have demonstrated that all intermediate nodes in $C$ must lie on a path in $C$ from some input in $I$ to some output in $O'$. Therefore, the no dangling nodes condition is satisfied.%
\end{proof}

% section the_preservation_of_informativeness (end)

When Theorem~\ref{thm:no_dangling_nodes_preserved} applies (i.e. the contiguousness requirements), we are able to demonstrate two ways informativeness is inherited: If the parents are partially informative, the child will be as well. And if the input parent is very informative, and the output parent is partially informative, the child will be very informative. The first theorem follows from the No Dangling Nodes condition both being inherited and implying partial informativeness. The second goes beyond that: it relies on the fact that the very informative input parent has a path from each input to the membrane, no matter where the membrane falls; so if the output parent provides a crossover compatible output part, then it must pick up each one of those paths and connect them to an output. 

To prove these theorems, we establish two lemmas.
 Lemma~\ref{lem:pathcut} points out that, if there is a path from an input to an output in some graph $G$, then that path must go through the membrane associated with any IO Partition of $G$. The intuition is straightforward: if the membrane is a fence between inputs from outputs, then a path from an input to an output must cross that fence. Formally:
\begin{lemma}\label{lem:pathcut}
	If IOD Graph $G$ satisfies the no dangling nodes condition and has at least one intermediate node, then for every path $p$ from an input to an output, and for any IO Partition $(\inpart{},\outpart{})$, there exists some forward edge $f$ in $p$ such that $f\in F(\inpart{},\outpart{})$.
\end{lemma}
\begin{proof}\label{pf:pathcut}
	Since there is a path from some input $i$ to some input $o$, there must be some node $n$ on the path in $\inpart{}$ which connects to some node on the path $n'$ in $\outpart{}$ (note, we are allowing $n=i$ or $n'=o$). Then edge $(n,n')$ is a forward edge in $F(\inpart{},\outpart{})$.
\end{proof}%testtesttest

Lemma~\ref{lem:membranepath} points out that no dangling nodes means that the forward links contained in any membrane must lie on a path from an input to an output. That is, if every step is a step on a path between an input and an output, then any step across the fence must be part of such a path. 
\begin{lemma}\label{lem:membranepath}
	If IOD Graph $G$ satisfies the no dangling nodes condition and has at least one intermediate node, then for any IO Partition $(\inpart{},\outpart{})$, every forward edge $f\in F(\inpart{},\outpart{})$ must lie on a path from the input set to the output set and the set of forward edges $F(\inpart{},\outpart{})$ must be non-empty.
\end{lemma}
\begin{proof}\label{pf:membranepath}
	By Lemma~\ref{lem:pathcut}, $\exists f \in F(\inpart{},\outpart{})$. Now, $f$ must connect two nodes $n,n'$ in $G$. By the no dangling nodes condition, there must be a path $p=(i,\ldots,n,\ldots,o)$ and $q=(i',\ldots, n', \ldots, o')$ for $i,i'\in I$, $o,o'\in O$. Then we can construct the path $r=(i,\ldots,n,n',\ldots,o')$. $f$ is on that path and connects the input set to the output set.
\end{proof}

This is a critical component in the retention of informativeness, because if there were available forward links which led to blind alleys, then a path from an input could be directed into a blind alley, as seen in Figures 
\ref{fig:IODGraphFullyToNonInform} and
\ref{fig:dangles}.

Now we are equipped to prove Theorem~\ref{thm:PartialInform}, which shows the inheritance of partial informativeness, and Theorem~\ref{thm:VeryInform}, which shows that a very informative input parent and partially informative output parent yield very informative children.

\begin{theorem}\label{thm:PartialInform}
	Let IOD Graphs $G$ and $G'$ be partially informative and satisfy the no dangling nodes condition. Then any crossover child produced by an input-contiguous IO partition of input parent $G$ and an output-contiguous IO partion of output parent $G'$ must  be partially informative.
\end{theorem}

\begin{proof}\label{pf:Partial}
	By Theorem~\ref{thm:no_dangling_nodes_preserved}, the child $C$ must satisfy the no dangling nodes condition. Since they are partially informative, $G$ and $G'$ each have at least one path from their input set to their output set, which we will name $p=(i, \ldots, o) \in N(G)$ and $p'=(i',\ldots, o') \in N(G')$. By Lemma~\ref{lem:membranepath}, every edge in $M(\inpart{},\outpart{})$ is on a path from $I$ to $O$ and every edge in $M(\inpart{}',\outpart{}')$ is on a path from $I$ to $O$. Therefore, every edge in any crossover membrane $\hat{M}$ must be on a path from $I$ to $O$.
\end{proof}

\begin{theorem}\label{thm:VeryInform}
	Suppose input parent IOD Graph $G$ is very informative and output parent $G'$ is partially informative, and both satisfy the no dangling nodes condition. 
	 Then any crossover child produced by an input-contiguous IO partition of input parent $G$ and an output-contiguous IO partion of output parent $G'$ must be very informative.
\end{theorem}
\begin{proof}\label{pf:Very}
	Since $G$ is very informative, for each $i\in I$, $\exists$ a path $p(i,\ldots,o_i)$ for some $o_i \in O$. By Lemma~\ref{lem:pathcut}, each path must have an edge in $M(\inpart{},\outpart{})$. By Lemma~\ref{lem:membranepath}, every edge must lie on a path between the input set and the output set. Therefore, in $C$, every $i$ must lie on a path to the output set. Therefore $C$ is very informative.
\end{proof}

\medskip

However, the inheritance of full informativeness proves difficult to guarantee. In particular, the no dangling nodes condition and contiguous IO partitions are insufficient to assure the preservation of full informativeness. The intuition is as follows: Consider some input. Now,  given full informativeness, there must be a path from that input to every output. However, a well-constructed crossover can direct all those paths to the \emph{same} output. This breaks full informativeness. (However, these fully informative parents will yield a very informative child by Theorem~\ref{thm:VeryInform}.) Theorem~\ref{thm:FullyInform} demonstrates how this remapping could occur:
\begin{theorem}\label{thm:FullyInform}
	Full informativeness of the parents is not necessarily preserved even if both parents satisfy the no dangling nodes condition and the IO partitions are contiguous.
\end{theorem}
\begin{proof}\label{pf:Fully}
		Suppose IOD Graphs $G$ and $G'$ are fully informative and satisfy the no dangling nodes condition with contiguous IO partitions $(\inpart{},\outpart{})$ and $(\inpart{}',\outpart{}')$. Suppose they both have the name number of inputs $J$ and number of outputs $K$ and suppose $J=K$. 
		
		Suppose that all paths from inputs to outputs are unique and suppose they have no intermediate nodes in common. Name these paths $p(i,o)$ in $G$ and $p'(i',o')$ in $G'$.
		
		Then construct a crossover membrane $\hat{M}$ according to the following algorithm:
		\begin{enumerate}
			\item Choose, without replacement, $i\in I$, and $o'\in O'$.
			\item For each $o\in O$ and $i' \in I'$, find the edge $e\in p(i,o)$ and $e' \in p'(i',o')$ such that $e\in M(\inpart{},\outpart{})$ and $e'\in M(\inpart{}',\outpart{}')$. Construct the edge $e''$ which connects the source of $e$ with the destination of $e'$. Add $e''$ to $\hat{M}$.
			\item Repeat until the sets $I$ and $O'$ have been exhausted.
		\end{enumerate}
	In the implied child IOD Graph $C$, by construction, each $i \in I$ has $J=K$ paths to one output $o' \in O'$ and no other output. Therefore $C$ is not fully informative.
\end{proof}

\section{Actionability} % (fold)
\label{sec:actionability}
As described above, informativeness characterizes how information flowing into  the system is utilized, in the sense that it characterizes how many inputs are eventually connected to an output. A symmetric concept is \emph{actionability}, which characterizes the amount of informed behavior flowing out of a system, in the sense that it characterizes how many outputs are connected \emph{from} an input. The levels of actionability mirror those of informativeness:
 A \textit{partially actionable IOD Graph} has at least one path from an input to an output, 
 a \textit{very actionable IOD Graph} has a path from some input to every output, 
 and a \textit{fully actionable IOD Graph} has a path from every input to every output. 
 On the other extreme, a \textit{non-actionable IOD Graph} has no paths from inputs to outputs. 	
 
Obviously, actionability and informativeness are very closely related. 
Indeed, as is apparent from the definition,
IOD Graphs are
 \emph{non-informative} if and only if they are \emph{non-actionable}, 
are \emph{partially informative} if and only if they are \emph{partially actionable}, and 
are \emph{fully informative} if and only if they are \emph{fully actionable}. They differ at the level of \emph{very}, in that IOD Graphs can be \emph{very informative} without being \emph{very actionable} and vice versa.
 
The symmetry of these concepts allows the immediate extension of Theorems~\ref{thm:PartialInform}, \ref{thm:VeryInform}, and \ref{thm:FullyInform} to actionability:
\begin{theorem}\label{thm:PartialAction}
	Let IOD Graphs $G$ and $G'$ be partially actionable and satisfy the no dangling nodes condition. 
	Then any crossover child produced by an input-contiguous IO partition of input parent $G$ and an output-contiguous IO partion of output parent $G'$ must be partially actionable.
\end{theorem}

\begin{theorem}\label{thm:VeryAction}
	Suppose input parent IOD Graph $G$ is partially actionable and output parent $G'$ is very actionable, and both satisfy the no dangling nodes condition. 
	 Then any crossover child produced by an input-contiguous IO partition of input parent $G$ and an output-contiguous IO partion of output parent $G'$ must be very actionable.
\end{theorem}

\begin{theorem}\label{thm:FullyAction}
	The no dangling nodes condition and contiguous IO partitions are insufficient to assure the preservation of full actionability. 
\end{theorem}

These theorems are presented without proof.

\section{Discussion}\label{sec:discussion}

IOD Graphs are a broad class of graphs that can be used to solve problems, and this paper defines a crossover operation on these graphs which can be used for evolutionary computation. 
 The utility of IOD Graphs rests on paths from inputs to outputs, so we define a measure of that connectedness through \emph{informativeness} (and \emph{actionability.})
 We show some levels of informativeness will be preserved under crossover under certain conditions, the most important of which is the \emph{no dangling nodes condition} which rules out blind alleys in the IOD Graph. 

There are several reasons why fully informative IOD Graphs may not be optimal for particular problems. Here are three possibilities. First, as mentioned elsewhere, it may be costly to maintain connections. Second, in the context of e.g. a municipal water system, it may not be best for every input--some fresh water, some grey water--to connect to every output--some local waterways, some a water reclamation plant. Third, suppose the IOD Graph is a decision engine, in which each output triggers a specific action, all of which are known to be sometimes useful. Suppose there are a large number of inputs, many of which are known to be noise. In this case, it may be optimal to have a partially informative and very actionable IOD Graph so that noise is ignored but all actions are available.

One aspect of crossover in this framework is that crossover membranes are not unique; that is to say, given an input part \inpart{} and an output part \outpart{}, there can, in general, be multiple crossover membranes which can yield a child. Moreover, these children might vary in their informativeness. This is known in the literature; \citet{schaffer1992combinations} named this the ``competing conventions problem,'' in the sense that the internal meaning of nodes is contextual and dependent on the particular network structure, and there is nothing which forces that meaning to be consistent across networks. So two parents could have different ``conventions'' with regard to mappings of inputs to outputs which would undermine the effectiveness of the crossover child. Consider 
Figure~\ref{fig:IODGraphCompeteConven}, which depicts a competing conventions problem. Suppose the correct solution is a path from I1 to O1 and a path from I2 to O2.
In this case, both parents solve the problem correctly while the child gets it exactly wrong. From a ``competing conventions'' perspective, this means that the ``convention'' or internal representation in nodes of these paths was not consistent/established across parents, leading to a breakdown during crossover.\footnote{Note, there is also a crossover membrane which produces a child which gets the problem exactly right, namely the crossover membrane $\left\{(N1, N4), (N2, N3)\right\}$.}

	\begin{figure}[h]
	  \centering
	    \includegraphics[width=.7\textwidth]{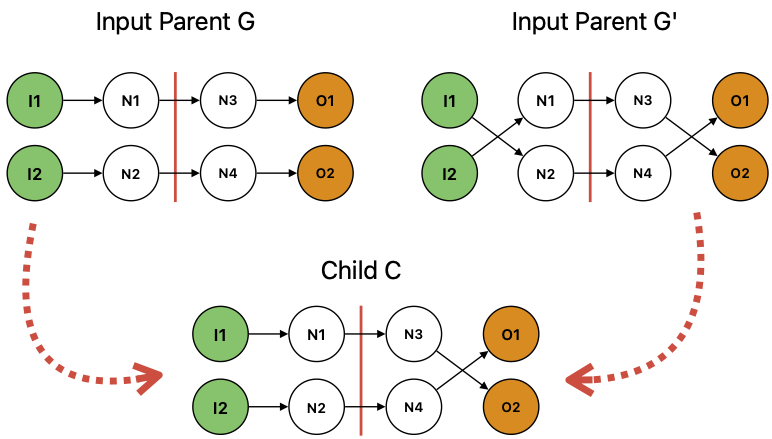}
	  \caption{Crossovers may suffer from a\\competing conventions problem
	  %$G$ and $G'$ solve the problem while $C$ fails
	  }
	  \label{fig:IODGraphCompeteConven}
	\end{figure}

Along these lines, while crossover can result in a non-informative child from fully informative parents, it's also the case that the reverse is also possible. For example, Figure~\ref{fig:IODGraphNonToFullyInform} depicts two non-informative parents with a fully informative child.

\begin{figure}[h]
  \centering
    \includegraphics[width=.7\textwidth]{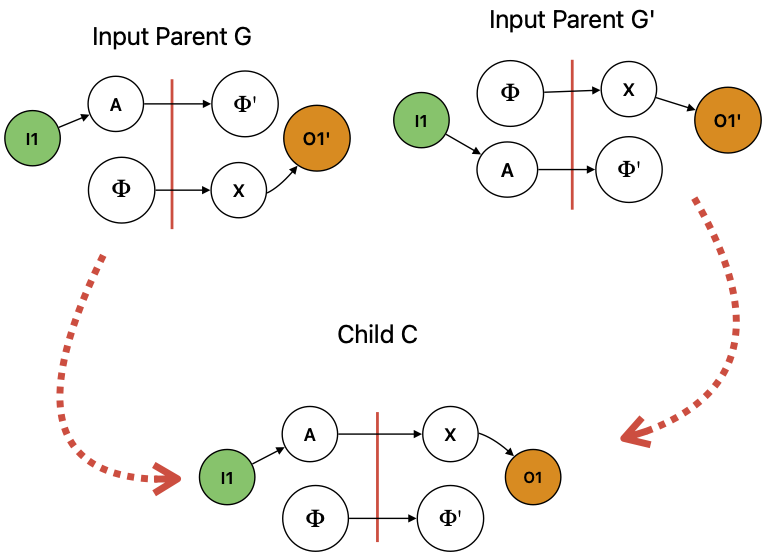}
  \caption{Two non-informative parents\\ can have a fully-informative child.}
  \label{fig:IODGraphNonToFullyInform}
\end{figure}

These two observations combined suggests that the selection of the crossover membrane will be very important as we consider the implementations of in evolutionary computation. One possible algorithmic solution is  to augment the crossover operation with endogenously ``learned'' crossover membranes or extend similarity-based methods such as \citet{dragoni2014simba} to IOD Graphs. We intend to pursue this in future work.

The possibly profound implications of the competing conventions problem also suggests that traditional neural network updating (i.e. excluding crossover) might perform well  precisely because, without crossover, there is a single internal ``convention'' of connections/weights that are encoded in the neural network at any particular point in time. There is no \textit{competing} convention, which would make learning muddied.

In some applications, there are particular kinds of nodes that can or cannot be connected. For example, consider a municipal system of flow delivery pipes, which contain either water or natural gas. One would not want to consider crossovers which connect a water pipe to a natural gas receptor or vice versa. Crossover as defined here can be easily generalized to include this case; nodes could be `tagged' and only connections between nodes of the same `tag' could be considered in the crossover membranes.

Not all networks of interest are IOD Graphs. Social networks, for example, are not IOD Graphs; in social networks, each node is both an input and an output, which is explicitly ruled out in IOD Graphs (where inputs and outputs are disjoint sets). Therefore, none of the properties attributed to IOD Graphs, such as those established in this study, necessarily apply to e.g. social networks. The generalization of the concepts and properties established in this paper to those or other networks is left for future work.

%\begin{figure}[h]
%    \centering
%        \caption{Distribution of IOD Graph Informativeness}
%    \begin{subfigure}[t]{0.95\textwidth}
%        \centering
%    \includegraphics[width=1\textwidth]{image-informativeness-degree-graph.png}
%	\caption{Informativeness by Network Density}
%  \label{fig:images_informativeness-degree-graph}
%    \end{subfigure}
%    \\
%    \vspace{1em}
%        \begin{subfigure}[t]{0.45\textwidth}
%        \centering
%    \includegraphics[width=.9\textwidth]{IODGraphDistributionExample.png}
%    \caption{An Example IOD Graph}
%  \label{fig:images_IODGraphDistributionExample}
%    \end{subfigure}
%\end{figure}

\begin{figure}[h]
    \centering
    \includegraphics[width=1\textwidth]{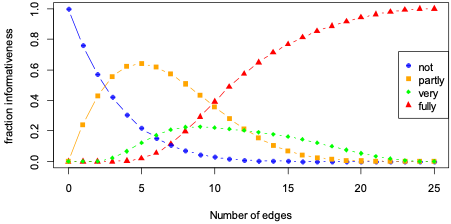}
        \caption{Distribution of IOD Graph Informativeness by Degree Density}
  \label{fig:images_informativeness-degree-graph}
\end{figure}

Figure~\ref{fig:images_informativeness-degree-graph} depicts the distribution of informativeness as the degree density of the underlying IOD Graphs varies. In particular, we consider the universe of all IOD Graphs with 3 inputs, 2 outputs, 5 intermediate nodes, supposing that each input is connected to one node and each output is connected from one node. Then 
Figure~\ref{fig:images_informativeness-degree-graph} depicts, for each number of intermediate node connections, what fraction are Not, Partially, Very, and Fully informative. As we can see, non-informative graphs disappear quickly as degree density increases, but, on the other hand, very informative graphs persist to a fairly high level of degree density. If edges are somehow `costly' to maintain, one might expect an intermediate level of degree density, where partially and very informative graphs together dominate.\footnote{Note that the number of IOD Graphs for each number of edges varies widely, in particular, exponentially. That is, while there is exactly one graph with zero edges, and exactly one graph with 25 edges, there are over 3 million IOD Graphs with 13 edges.}

Informativeness is only one way to measure the connectedness from inputs to outputs; in particular, one could consider metrics describing more subtle measures of connectedness, which count number of paths among possible paths, which could break down not, partial, very, and fully informative graph categories. The analysis of these metrics will be considered in future work.

\begin{figure}[p]
    \centering
    \begin{subfigure}[t]{0.9\textwidth}
        \centering
\begin{tabular}{ccc}
	\includegraphics[width=.22\textwidth]{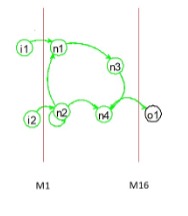} &
	\includegraphics[width=.22\textwidth]{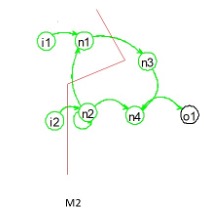} &
	\includegraphics[width=.22\textwidth]{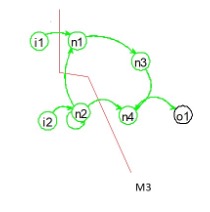} \\
	\includegraphics[width=.22\textwidth]{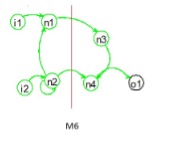} &
	\includegraphics[width=.22\textwidth]{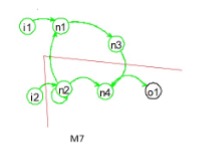} &
	\includegraphics[width=.22\textwidth]{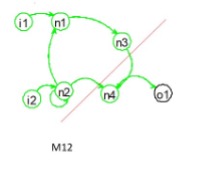}
\end{tabular}
        \caption{Seven Contiguous IO Partitions}
		\label{fig:contiguous}
    \end{subfigure}
	~
	\vspace{3em}
	
    \begin{subfigure}[b]{0.9\textwidth}
        \centering
\begin{tabular}{ccc}
	\includegraphics[width=.22\textwidth]{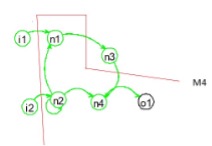}  &
	\includegraphics[width=.22\textwidth]{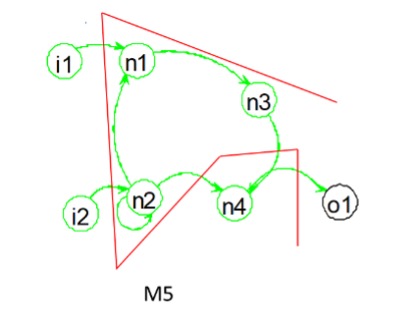}  &
	\includegraphics[width=.22\textwidth]{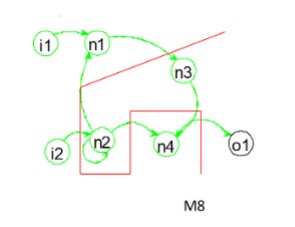}  \\
	\includegraphics[width=.22\textwidth]{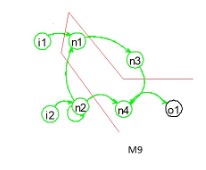}  &
	\includegraphics[width=.22\textwidth]{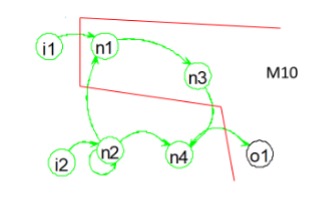} &
	\includegraphics[width=.22\textwidth]{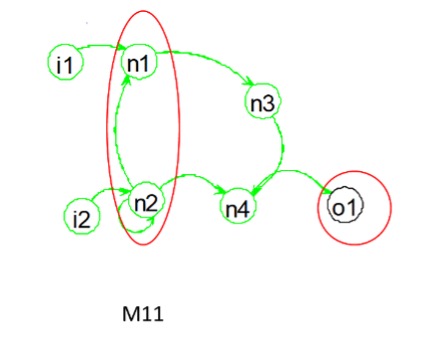} \\
	\includegraphics[width=.22\textwidth]{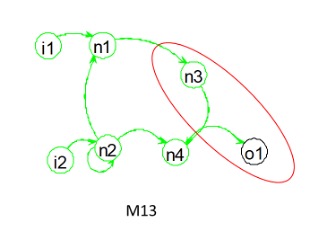} &
	\includegraphics[width=.22\textwidth]{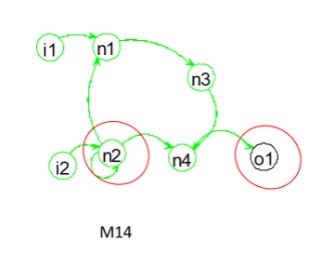} &
	\includegraphics[width=.22\textwidth]{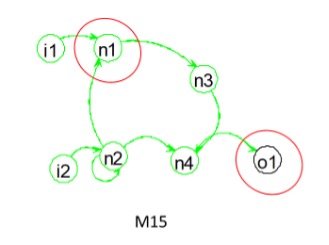}
\end{tabular}
        \caption{Nine Non-contiguous IO Partitions}
		\label{fig:notcontiguous}
    \end{subfigure}%
    \caption{All IO Partitions of an example IOD Graph
	}\label{fig:computational_parts}
\end{figure}

\subsection{Evo-Devo} % (fold)
\label{sub:evo_devo}

Evo-devo (short for ``Evolutionary Development'')
is a gene-driven process in which a genome representation is translated into phenotypes, in this case, graphs, and crossover is done on the genomes, not the graphs themselves
 \citep{jacob1977evolution,gould1977ontogeny}.
IOD Graphs are the phenotype of this process.
The genomes in evo-devo, which can have any form, are typically vectors of parameters.
Evo-devo
is common in applications of evolutionary methods to IOD Graphs \cite[e.g.][]{arifovic2001using, garcia2006alternative}.
Some studies, such as 
\citet{dragoni2014simba} and \citet{uriot2020safe}, 
apply crossover to genomes in a network-aware method. In an approach most similar to ours, \citet{uriot2020safe}
build a representation 
between hidden layers of feed-forward neural networks based on functionality and applies crossover to that representation, while our crossover allows for arbitrary crossover cuts and for a larger class of graphs. 

By contrast with evo-devo, the crossover process described here is directly applied on the network topology and therefore is agnostic with respect to the genome representation. IOD Graph Crossover  could lead to improved performance of evolution if the `direct' representation of the graph fosters Lamarckian evolution (in which acquired traits can be inherited to offspring) through maintaining coherent sub-structures. This could be important for e.g. modifying trained neural networks. 

Many existing evo-devo processes will likely not achieve offspring consistent with IOD Graph Crossover 
as described in Definition \ref{def:IODGraphCrossover}. However, 
 if a particular evo-devo process \emph{could} be shown to obey all the properties in Definition \ref{def:IODGraphCrossover}, 
then such a process would therefore follow all theorems here with regard to e.g. informativeness and actionability preservation, as well as any future theorems for IOD Graph crossover. 
By contrast, the crossover method in the study \citet{uriot2020safe}, described above, does not require  preserving the number of links at the crossover as we require for crossover compatibility (Definition \ref{def:crossover_compatible}, Section~\ref{sec:the_crossover_operation}). Therefore the results in this paper do not apply.

 Some studies apply crossover only to the weights of the neural networks, not the graph structure \cite[e.g.][]{braun1993evolving, braunuser, chandra2012crossover}. These are not examples of IOD Graph Crossover, which only considers crossover on structure, not weights.

% subsection evo_devo (end)

\section{Conclusion} % (fold)
\label{sec:conclusion}

Here we defined Input/Output Directed Graphs to capture a broad class of networks which connect inputs to outputs. This class includes feed-forward neural networks and perceptrons, electrical circuits, municipal water systems, chemical reaction networks, and data flow networks. Various studies have applied evolutionary methods to optimizing these networks in their distinct domains. 
The goal of this study is to provide a strong theoretical foundation and common framework for crossover across all types of IOD Graphs and connect crossover as defined here with our proposed measure of informativeness (and actionability), which provides a measure of the connectedness of inputs to outputs in IOD Graphs. 
There are many directions for future work, e.g. generalized methods of constructing crossover membranes, results for IOD Graphs with nodes of different types or other parameters, generalization of informativeness to other types of graphs, and more nuanced, possibly continuous, measures of informativeness and actionability. 
One area of future work is a computational implementation of crossover on general IOD Graph structures. 
Figure~\ref{fig:computational_parts} shows output of one such computational implementation we are developing.
 In this Figure, all sixteen IO partitions of a particular IOD Graph are generated. We identify contiguousness and construct crossover membranes computationally. A general IOD Graph framework such as this will support the integration of a variety of  methods and applications from different domains and support general evolutionary computation on IOD Graphs.

% section conclusion (end)

%%%%%%%%%%%%%%%%%%%%%%%%%%%%%%%%%%%%%%%%%%%%%%%%%%%%%% Bibliography
\newpage
	\bibliographystyle{apalike}
	\bibliography{PapeSchafferSayamaZosh_IODGraphsInformativeness}
	%\bibliography{test}
	
%%

\typeout{get arXiv to do 4 passes: Label(s) may have changed. Rerun}

\end{document}